\documentclass[letterpaper, 10 pt, conference]{ieeeconf}  % Comment this line out
\IEEEoverridecommandlockouts                              % This command is only
\usepackage{tikz}
\usepackage{subcaption}
\usepackage{booktabs}
\usepackage{url}
\usepackage{cite}
\usepackage{amsmath,amssymb,amsfonts}
\usepackage{algorithmic}
\usepackage{graphicx}
\usepackage{textcomp}
\usepackage{xcolor}
\def\BibTeX{{\rm B\kern-.05em{\sc i\kern-.025em b}\kern-.08em
    T\kern-.1667em\lower.7ex\hbox{E}\kern-.125emX}}
\newtheorem{prop}{Proposition}%[section]
\newtheorem{prob}{Problem}%[section]
\newtheorem{lem}{Lemma }%[section]s
\newtheorem{defn}{Definition}%[section]
\newtheorem{rem}{Remark }%[section]
\usepackage[ruled,vlined]{algorithm2e}
\usepackage{makecell}

\newcommand{\Real}{\mathbb R}
\newcommand{\eps}{\varepsilon}

\newcommand{\norm}[1]{\left\Vert#1\right\Vert}
\newcommand{\ra}{\rightarrow}
\newcommand{\set}[1]{\left\{#1\right\}}

\newcommand{\D}{\mathcal{D}}

\let\subset\subseteq

\title{\LARGE \bf
Compositionally Verifiable Vector Neural Lyapunov Functions for Stability Analysis of Interconnected Nonlinear Systems}

\author{Jun Liu, Yiming Meng, Maxwell Fitzsimmons, and Ruikun Zhou % <-this % stops a space
\thanks{This research was supported in part by an NSERC Discover Grant and the Canada Research Chairs program. This research was enabled in part by support provided by the Digital Research Alliance of Canada (alliance.ca).}% <-this % stops a space
\thanks{Jun Liu, Maxwell Fitzsimmons, and Ruikun Zhou are with the Department of Applied Mathematics, Faculty of Mathematics, University of Waterloo, Waterloo, Ontario N2L 3G1, Canada.  Emails: \texttt{j.liu@uwaterloo.ca, mfitzsimmons@uwaterloo.ca, ruikun.zhou@uwaterloo.ca}
        }%
\thanks{Yiming Meng is with  the
Coordinated Science Laboratory, University of Illinois Urbana-Champaign,
Urbana, IL 61801, USA. Emails: 
        \texttt{ymmeng@illinois.edu, yiming.meng@uwaterloo.ca}}%
}

\begin{document}

\maketitle
\thispagestyle{empty}
\pagestyle{empty}

%%%%%%%%%%%%%%%%%%%%%%%%%%%%%%%%%%%%%%%%%%%%%%%%%%%%%%%%%%%%%%%%%%%%%%%%%%%%%%%%
\begin{abstract}
While there has been increasing interest in using neural networks to compute Lyapunov functions, verifying that these functions satisfy the Lyapunov conditions and certifying stability regions remain challenging due to the curse of dimensionality. In this paper, we demonstrate that by leveraging the compositional structure of interconnected nonlinear systems, it is possible to verify neural Lyapunov functions for high-dimensional systems beyond the capabilities of current satisfiability modulo theories (SMT) solvers using a monolithic approach. Our numerical examples employ neural Lyapunov functions trained by solving Zubov's partial differential equation (PDE), which characterizes the domain of attraction for individual subsystems. These examples show a performance advantage over sums-of-squares (SOS) polynomial Lyapunov functions derived from semidefinite programming. 
\end{abstract}

%%%%%%%%%%%%%%%%%%%%%%%%%%%%%%%%%%%%%%%%%%%%%%%%%%%%%%%%%%%%%%%%%%%%%%%%%%%%%%%%
\section{Introduction}

One of the longstanding challenges in nonlinear control is the construction of Lyapunov functions for stability analysis and controller design. Since Lyapunov's seminal paper over a century ago \cite{lyapunov1992general}, there has been an ongoing quest for constructive methods for the design of Lyapunov functions. Both analytical \cite{haddad2008nonlinear,sepulchre2012constructive} and computational \cite{giesl2007construction, giesl2015review} strategies have been explored.

Among computational approaches for finding Lyapunov functions, sums-of-squares (SOS) techniques have perhaps received the most attention \cite{papachristodoulou2002construction,papachristodoulou2005tutorial,packard2010help,tan2008stability,topcu2008local,jones2021converse}. Not only can these techniques provide local stability analysis, but they can also offer regions of attraction estimates \cite{topcu2008local,tan2008stability,packard2010help}. Using semidefinite programming (SDP), the region of attraction can be enlarged by expanding the level set of a certain ``shape function" contained in the region of attraction estimate. While this is computationally appealing, the selection of such shape functions beyond the obvious choices of norm functions remains somewhat ad hoc \cite{khodadadi2014estimation}. 

In recent years, progress in machine learning and neural networks has begun to transform the realm of computational studies. Many authors have investigated the use of neural networks for computing Lyapunov functions (see, e.g., \cite{grune2021computing,dai2021lyapunov,zhou2022neural,abate2020formal,gaby2022lyapunov,kang2021data}, and \cite{dawson2023safe} for a recent survey). Perhaps one of the most notable differences between SDP-based synthesis of SOS Lyapunov functions and the training of neural Lyapunov functions is that neural network Lyapunov functions are not always guaranteed to be Lyapunov functions. Subsequent verification is required to ensure they satisfy Lyapunov conditions, e.g., using satisfiability modulo theories (SMT) solvers \cite{chang2019neural,ahmed2020automated}. This process is time-consuming, especially when seeking a maximal Lyapunov function \cite{liu2023towards}, as accurate numerical verification becomes more challenging near the stability region's boundary. A neural network Lyapunov function offers the advantage of universal approximation and embracing non-convex optimization, unlike the SDP's limitation to convex optimization. It can sometimes capture regions of attraction better and is capable of dealing with non-polynomial dynamics. Moreover, it can leverage the abundant machine learning infrastructure, including neural network architectures, optimization algorithms, and graphics processing unit (GPU) computing. Overcoming the curse of dimensionality in verifying neural Lyapunov functions remains a significant challenge.

In this work, we demonstrate that by exploiting the compositional structure of interconnected nonlinear systems,  where each subsystem admits a neural Lyapunov function in the absence of interaction, one can verify a vector neural Lyapunov function for the interconnected high-dimensional system to determine a region of attraction. This approach can address problems beyond the capabilities of current SMT solvers \cite{gao2013dreal} using a monolithic approach. Specifically, we conduct numerical experiments using neural Lyapunov functions trained by solving Zubov's partial differential equation (PDE), which characterizes the domain of attraction for individual subsystems \cite{zubov1964methods,kang2021data,liu2023towards}. We show that these functions outperform SOS polynomial Lyapunov functions obtained through semidefinite programming in terms of capturing the region of attraction of the interconnected system. This is an extended version of the paper presented in \cite{liu2023compositionally_acc}, with an additional appendix for omitted technical proofs. 

\section{Preliminaries}

\subsection{Interconnected system}

Consider a network of nonlinear systems of the form
\begin{equation}
    \label{eq:sys}
    \dot x_i  = f_i(x_i) + \sum_{j\neq i}  G_{ij}(x_i,x_j),
\end{equation}
where each $f_i:\,\Real^{n_i}\ra\Real^{n_i}$ and $G_{ij}:\,\Real^{n_i}\times \Real^{n_j}\ra\Real^{n_i}$ are assumed to be locally Lipschitz and $i,j\in \set{1,\ldots,l}$, which indicates the collection of subsystems in the network. It is assumed that $f_i(0)=0$ and $G_{ij}(0,0)=0$. Hence, the origin is an equilibrium point for each individual subsystem in the absence of the interconnection terms $G_{ij}$ and for the overall interconnected system. 

We use $x=(x_1,\ldots,x_l)\in \Real^{n_1}\times\cdots \times\Real^{n_l}=\Real^N$, where $N=n_1+ n_2 + \cdots + n_l$, to denote the state vector of the networked system (\ref{eq:sys}) and $\phi(t,x)$ to indicate the solution to (\ref{eq:sys}) starting from the initial condition $\phi(0)=x$. We also use $\phi_i(t,x)$ to denote the solution trajectory of the $i$th subsystem. 

\begin{defn}[Domain of Attraction]
    Suppose that the origin is asymptotically stable for (\ref{eq:sys}), the \textit{domain of attraction} of the origin for (\ref{eq:sys}) is defined as
    $$
    \D: = \set{x\in\Real^N:\,\lim_{t\ra \infty} \phi(t,x) = 0}.   
    $$ 
    Any invariant subset of $\D$ is called a \textit{region of attraction}. 
\end{defn}

\subsection{Problem formulation}

A common approach to computing regions of attraction is using sub-level sets of Lyapunov functions. Suppose that we can compute a Lyapunov function for each individual subsystem in the absence of the interconnection terms $G_{ij}$, namely subsystems of the form
\begin{equation}
    \label{eq:subsys}
    \dot x_i  = f_i(x_i),\quad i=1,\ldots, l,
\end{equation}
such that sub-level sets of the form
\begin{equation}
    \label{eq:omegai}
\mathcal{V}_i(v_i):=\set{x_i\in \Real^{n_i}:\, V_i(x_i)\le v_i}
\end{equation}
are verified regions of attraction for individual subsystems in (\ref{eq:subsys}), where $v_1,\ldots,v_l$ are positive constants. 

Clearly, the absence of the interconnection, i.e., $G_{ij}\equiv 0$, implies that the set $\Omega_1\times \cdots \times \Omega_l$ is a region of attraction for (\ref{eq:sys}). However, with interconnection, this is no longer the case. The main objective of this paper is to formulate compositionally verifiable conditions to certify regions of attraction for the interconnected system (\ref{eq:sys}).

\begin{prob}
Given regions of attraction for individual subsystems provided by (\ref{eq:omegai}), verify regions of attraction for the interconnected system (\ref{eq:sys}). 
\end{prob}

% \subsubsection*{Notation} 

\section{Stability and reachability analysis using vector Lyapunov functions}\label{sec:analysis}

Consider a vector Lyapunov function $V:\,\Real^N\ra \Real^l$ by 
\begin{equation}
    \label{eq:V}
V(x)=(V_1(x_1),\ldots,V_l(x_l)),
\end{equation}
where $V_i:\,\Real^{n_i}\ra\Real$ are scalar Lyapunov functions for subsystems in (\ref{eq:subsys}).

\subsection{Local stability analysis}

We assume that each $f_i$ is continuously differentiable and the origin is exponentially stable for each individual subsystem (\ref{eq:subsys}). Let $A_i = D f_i(0)$, where $Df_i$ is the Jacobian of $f_i$. By our assumption, $A_i$ is a Hurwitz matrix. Let $Q_i>0$ be any positive definite matrix and let $P_i>0$ solve the Lyapunov equation 
\begin{equation}
    \label{eq:lyap}
    P_iA_i +A_i^T P_i = -Q_i. 
\end{equation}
Rewrite (\ref{eq:subsys}) as 
\begin{equation}
    \label{eq:subsys2}
    \dot x_i = A_ix_i + g_i(x_i),
\end{equation}
where $g_i(x_i)=f_i(x_i)-A_ix_i$ satisfies $\lim_{x_i\ra 0}\frac{\norm{g_i(x_i)}}{\norm{x_i}}=0$.

We can analyze local stability of the interconnected system (\ref{eq:sys}) using the following proposition. 

\begin{prop}\label{prop:local}
Let $P_i$ and $Q_i$ satisfy (\ref{eq:lyap}). For each $i=1,\ldots,l$ and $p>0$, define the set 
$$
\mathcal{P}_i(p) := \set{x\in\Real^{n_i}:\,x^TP_ix\le p}. 
$$
Suppose that there exists a positive vector $c=(c_1,\ldots,c_l)\in\Real^l$ and a matrix of nonnegative elements $R=(r_{ij})$ such that the following inequalities hold:
\begin{align}
\norm{P_i D g_i(x_i)} &\le r_{ii}, \label{eq:ineq1}\\
\norm{P_i D G_{ij}(x_i,x_j)} &\le r_{ij}, \label{eq:ineq2}
\end{align}
for all $x_i\in \mathcal{P}_i(c_i)$ and $x_j\in \mathcal{P}_j(c_j)$.  
Define $\Lambda = (\lambda_{ij})$ by 
\begin{align}
\lambda_{ii} &= - \frac{\lambda_{\min}(Q_i)}{\lambda_{\max}(P_i)} + \frac{2(r_{ii} + \sum_{j\neq i} r_{ij})}{\lambda_{\min}(P_i)},\label{eq:gain1}\\
\lambda_{ij} &= \frac{r_{ij}}{\lambda_{\min}(P_j)},\quad j\neq i. \label{eq:gain2}
\end{align}
If there exists a positive vector $p=(p_1,\ldots,p_l)\in \Real^l$ such that $p\le c$ and $\Lambda p < 0$ (componentwise), then the set
$$
\mathcal{P}(p) := \mathcal{P}_1(p_1)\times \cdots \times \mathcal{P}_l(p_l), 
$$
is a region of attraction for the interconnected system (\ref{eq:sys}). 
\end{prop}

\begin{rem}
While we assumed that $G_{ij}$ takes the form in (\ref{eq:sys}), a more general form of interconnection, $G(x)$, is permissible for the analysis in Proposition \ref{prop:local} to proceed. We opted for the summation form to exploit the compositional nature of (\ref{eq:ineq1}) and (\ref{eq:ineq2}). If a decomposition is available through other inequality estimates, the results of Proposition \ref{prop:local} can be applied straightforwardly. Moreover, the gain bounds (\ref{eq:gain1}) and (\ref{eq:gain2}) may be further improved (see the appendix).
\end{rem}

\begin{rem}\label{rem:m-matrix}
The condition that there exists some $p>0$ such that $\Lambda p<0$ or $(-\Lambda) p>0$ is one of many conditions equivalent to $-\Lambda$ being a nonsingular $M$-matrix (condition $K_{33}$ in \cite{plemmons1977m}), provided $-\Lambda$ is a $Z$-matrix (square matrices with nonpositive off-diagonal entries). Another equivalent condition is that $\Lambda$ is Hurwitz (condition $J_{29}$ in \cite{plemmons1977m}).
\end{rem}

\begin{rem}
Vector Lyapunov functions and the comparison lemma have been well-known since the 1960s \cite{bellman1962vector,bailey1965application,matrosov1962theory}. Their application in analyzing large-scale systems was popularized in the 1970s and 80s. For interested readers, see a survey in \cite{michel1983status} and the books \cite{siljak1978large,vidyasagar1981input}. ROA estimates have received less attention, but prior work exists \cite{weissenberger1973stability,bitsoris1977stability}. Our contribution is to formulate this in a compositional form that is readily verifiable by numerical SMT solvers, such as dReal \cite{gao2013dreal}, even with numerical errors around the origin. Specifically, we examine the linearization around the origin and use the robustness implied by linear system stability to bound high-order terms around the origin. One advantage of using a vector Lyapunov function, as opposed to a weighted sum of scalar Lyapunov functions, for stability and reachability analysis is that it offers invariance certification for hyperrectangular sets (in $\mathbb{R}^l$) rather than ellipsoids or rhombi, a benefit noted in \cite{weissenberger1973stability}.
\end{rem}

\subsection{Reachability analysis}

The local analysis in the previous section is inherently conservative because of the use of the matrix norm and a quadratic Lyapunov function for an interconnected nonlinear system. The next result provides verifiable conditions to expand the region of attraction through reachability analysis using vector Lyapunov functions. 

\begin{prop}\label{prop:reach}
Suppose that there exist Lyapunov functions $V_i:\,\Real^{n_i}\ra\Real$, positive vectors 
$c=(c_{1},\ldots,c_{l})\in\Real^l$ and $v=(v_{1},\ldots,v_{l})\in\Real^l$, and a constant $\eps<0$ 
such that $c\le v$ and the following inequalities hold:
\begin{align}
\nabla V_i(x_i) \left [f_i(x_i) + \sum_{j\neq i}  G_{ij}(x_i,x_j) \right] \le \eps < 0 \label{eq:ineq3}
\end{align}   
for all $x_i \in \set{x_i\in\Real^{n_i}:\, c_{i}\le V_i(x_i) \le v_i }$ and all $x_j\in \mathcal{V}_j(v_j):=\set{x_j\in\Real^{n_j}:\, V_j(x_j) \le v_j }$. Then solutions of the interconnected system (\ref{eq:sys}) starting in $\mathcal{V}(v)= \mathcal{V}_1(v_1)\times \cdots \times \mathcal V_l(v_l)$ reach $\mathcal{V}(c)$ in finite time and remain there afterwards. 
\end{prop}

\section{Training vector neural Lyapunov functions and SMT verification}\label{sec:method}

In this section, we describe the proposed computational approach for training and verifying neural vector Lyapunov functions for interconnected systems. Both the learning and verification are conducted in a compositional fashion so that it can leverage the compositional structure of the interconnected system to scale the approach to tackle problems beyond the reach of monolithic approaches.

\subsection{Training neural Lyapunov functions for subsystems}\label{sec:method:train}

We following the approach presented in \cite{liu2023towards,liu2023physics,liu2024lyznet} for training neural Lyapunov functions for individual subsystems using Zubov's equation and physics-informed neural networks \cite{lagaris1998artificial,raissi2019physics}. The main idea is to solve Zubov's PDE 
\begin{equation}
    \label{eq:zubov}
    \dot W(x):= \nabla W(x) \cdot f(x) = - \Psi(x) (1-W(x)), 
\end{equation}
where $W$ is a positive definite function to be solved and $\Psi$ is also a positive definite function. Two particular choices \cite{liu2023towards, liu2023physics} of \(\Psi\) are \(\Psi_1(x) = \alpha \norm{x}^2\) or \(\Psi_2(x) = \alpha (1+W(x)) \norm{x}^2\) for some positive constant \(\alpha > 0\). For numerical examples in this paper, we choose \(\Psi(x) = \alpha (1+W(x)) \norm{x}^2\) with \(\alpha = 0.1\), which corresponds to taking a transform \(\beta(s) = \tanh(\alpha s)\) of a usual Lyapunov function \cite{kang2021data}. 

Let $W_{\text{NN}}(x;\theta)$ be a neural network solution to Zubov's PDE (\ref{eq:zubov}) on a domain $\Omega\subset\Real^N$. Consider the training loss 
\begin{equation}
    \label{eq:loss}
    \mathcal{L}(\theta) = L_r(\theta)  + L_b(\theta) + L_d(\theta), 
\end{equation}
where $L_r$ is the residual error of the PDE, evaluated over a set of collocation points $\set{x_i}_{i=1}^{N_r}\subset \Omega$ as
\begin{equation}
    \label{eq:L_r}
    L_r = \frac{1}{N_c}\sum_{i=1}^{N_c}(\nabla_{x} W_{\text{NN}}(x_i;\theta) f(x_i) + \Psi(x_i)(1-W_{\text{NN}}(x_i;\theta)))^2. 
\end{equation}
The loss $L_b$ captures boundary conditions. As in \cite{liu2023towards, liu2023physics}, we can encourage the following inequality near the origin:
\begin{equation} \label{eq:c1Wc2}
\beta(c_1\norm{x}^2) \le W(x) \le \beta(c_2 \norm{x}^2),
\end{equation}
where $\beta(s)=\tanh(\alpha s)$ as discussed above.

Finally, the term $L_d(\theta)$ can capture a data loss of the form 
\begin{equation}
    \label{eq:L_d}
    L_d(\theta) = \frac{1}{N_d}\sum_{i=1}^{N_d}(W_{\text{NN}}(y_i;\theta)-\hat W(y_i))^2,
\end{equation}
where $\{\hat W(y_i)\}_{i=1}^{N_d}$ is a set of data points, which can be obtained by forward integration of the ODE defining (\ref{eq:subsys}) \cite{kang2021data}.

\subsection{Verification of stability and reachability using vector neural Lyapunov functions and SMT solvers}\label{sec:method:verify}

We aim to use satisfiability modulo theories (SMT) solvers to verify the inequalities (\ref{eq:ineq1}), (\ref{eq:ineq2}), and (\ref{eq:ineq3}) in terms of quadratic Lyapunov functions and the trained neural Lyapunov functions to certify local stability and regions of attraction. The procedure proceeds as follows. 

\textbf{(1) Local stability: } We use quadratic Lyapunov functions, obtained by solving Lyapunov equations (\ref{eq:lyap}), to verify inequalities (\ref{eq:ineq1}) and (\ref{eq:ineq2}). The matrix $\Lambda$ is computed according to (\ref{eq:gain1}) and (\ref{eq:gain2}). The choice of initial $c$ is arbitrary, as long as it is sufficiently small such that the matrix $-\Lambda$ becomes a nonsingular $M$-matrix \cite{plemmons1977m}. This is always achievable because as $r_{ij}$ becomes sufficiently small, the matrix $\Lambda$ will be diagonally dominating. Once the $M$-matrix condition is met, there always exists some $p>0$ such that $\Lambda p<0$ \cite{plemmons1977m}. 

For simplicity, we can use the largest vector $c$ obtained by the conditions in \cite{liu2023towards,liu2023physics} for local stability of individual subsystems and scale it down until inequalities (\ref{eq:ineq1}) and (\ref{eq:ineq2}) are verified such that $\Lambda$ defined by (\ref{eq:gain1}) and (\ref{eq:gain2}) satisfies  $\Lambda c<0$. As a result, we have verified a local stability region 
$$
\mathcal{P}(p) := \mathcal{P}_1(p_1)\times \cdots \times \mathcal{P}_l(p_l), 
$$
where $\mathcal{P}_i(p_i) := \set{x_i\in\Real^{n_i}:\,x_i^TP_ix_i\le p_i}$, as defined in Proposition \ref{prop:local}.

\textbf{(2) Reachability by quadratic Lyapunov functions:} The local stability region can be very small due to the use of matrix norms for estimates. We can employ a successive procedure to enlarge the region of attraction by Proposition \ref{prop:reach} as follows. Suppose there exists a sequence of vectors $c^m_P$ ($m=1,2,\cdots,k$) such that Proposition \ref{prop:reach} holds with $V_i(x_i)=x_i^TP_ix_i$, $c=c^m_P$, and $v=c^{m+1}_P$ for each $m=1,\ldots, k-1$. It follows that if $c^1_P<p$, where $p$ is from \textbf{Step 1}, then the set $\mathcal{P}(c^k_P)$ is an enlarged region of attraction for the interconnected system (\ref{eq:sys}).

\textbf{(3) Reachability by general vector Lyapunov functions: } Following \textbf{Step 2}, we employ SMT solvers to verify that the neural vector neural Lyapunov $V(x)=(V_1(x_1),\ldots,V_l(x_l))$ trained according to Section \ref{sec:method:train} can satisfy Proposition \ref{prop:reach} with $c=c^m_V$, and $v=c^{m+1}_V$, along with a sequence of vectors $c^m_V$ ($m=1,2,\cdots,q$) such that the following set containment condition is met
$$
\mathcal{V}(c^1_V):=\set{x\in\Real^N:\, V(x)\le c^1_V} \subset \mathcal{P}(c^k_P), 
$$
where $c^k_P$ is from \textbf{Step 2}. Following this, we potentially obtain a further enlarged region of attraction $\mathcal{V}(c^q_V)$. 

\subsection{Compositional nature of SMT verification}

The training process is naturally compositional because the neural Lyapunov functions are trained for individual systems separately. In this subsection, we discuss how verification of the aforementioned Lyapunov conditions in terms of inequalities can also be verified in a compositional fashion. 

Local stability is done through verifying (\ref{eq:ineq1}) and (\ref{eq:ineq2}), which is compositional because (\ref{eq:ineq1}) only involves one subsystem and (\ref{eq:ineq2}) only involves two. By examining the nature of $G_{ij}$ in concrete examples, we may be able to further decompose the verification of (\ref{eq:ineq2}). 

The search for the sequences $\{c^m_P\}$ and $\{c^m_V\}$ in \textbf{Steps 2 and 3} of Section \ref{sec:method:verify} can be formulated as successively solving a minimization problem of the form 
\begin{equation}
\begin{aligned}
& \underset{c_i}{\text{minimize}}
& & c_i \\
& \text{subject to}
& & (\ref{eq:ineq3}),
\end{aligned}
\end{equation}
where $v$ is given. We can solve this problem for all $i=1,\ldots,l$ and then set $v=c$ and continue until the value of $c$ cannot be improved further by predefined  threshold. 

The constraint given by the inequality (\ref{eq:ineq3}) can be further decomposed. We can verify the existence of constants $g_{ij}$
\begin{equation}\label{eq:ineq4}
\nabla V_i(x_i) G_{ij}(x_i,x_j) \le g_{ij}
\end{equation}
for all $x_i\in \mathcal{V}_i(v_i)$ and $x_j\in \mathcal{V}_j(v_j)$ and replace (\ref{eq:ineq3}) with
\begin{equation}\label{eq:ineq5}
\nabla V_i(x_i) f_i(x_i) + \sum_{j\neq i} g_{ij} \le \eps < 0.
\end{equation}
The compositional nature of (\ref{eq:ineq4}) and (\ref{eq:ineq5}) is evident. Both of them can be verified readily by numerical SMT solvers. We use dReal \cite{gao2013dreal} for verification in the numerical examples in this paper due to its $\delta$-completeness guarantees. Here, $\delta$ is a user-defined precision parameter. Verification of a given formula is guaranteed to succeed if the ``$\delta$-weakening" of the negation of the formula is unsatisfiable.

\section{Numerical results}\label{sec:numerical}

\subsection{Networked Van der Pol oscillators}

Inspired by \cite{kundu2015sum}, we consider a network of reversed Van der Pol equations of the form
\begin{align*}
\dot x_{i1} &= -x_{i2}, \\
\dot x_{i2} &= x_{i1} - \mu_i (1 - x_{i1}^2)x_{i2} + \sum_{j\neq i} \mu_{ij} x_{i1}x_{j2},
\end{align*}
where $\mu_i\in (0.5, 2.5)$ and $\mu_{ij}$ represents the interconnection strength. The parameters $\mu_i$ are randomly generated and take the following values for $i=1,\ldots,10$:
$$
 [1.25, 2.4, 1.96, 1.7, 0.81, 0.81, 0.62, 2.23, 1.7, 1.92]. 
$$
We set that $\mu_{ij}\in (-0.1,0.1)$ and the number of nonzero entries in $\set{\mu_{ij}}$ for each $i$ is fewer than $3$ or $4$ (referred to as density below). We choose the number of total subsystems $l=10$. Two network topologies are depicted in Fig. \ref{fig:network} for density equal to 3 and 4. The total dimension of the interconnected system is therefore 20, which is beyond the capability of current SMT or SDP-based synthesis of Lyapunov functions if a monolithic approach is taken. Using the approach proposed in this paper, we are able to train neural vector Lyapunov functions and verify regions of attraction. Comparisons with sums-of-squares (SOS) synthesis demonstrate that the neural Lyapunov functions, computed by solving Zubov's equation, outperform the ``expanding interior'' approach taken by SOS design.

\begin{figure}[ht]
    \centering
    \begin{subfigure}{0.2\textwidth}
        \centering
        \begin{tikzpicture}[scale=0.9, transform shape]
            % Define the nodes
            \foreach \i in {1,...,10} {
                \node[circle, fill=blue!50] (N\i) at (360/10*\i:2) {\i};
            }
            
            % Draw the arrows based on the first matrix
    \draw[->] (N9) -- (N1);
    \draw[->] (N6) -- (N2);
    \draw[->] (N9) -- (N2);
    \draw[->] (N2) -- (N3);
    \draw[->] (N10) -- (N3);
    \draw[->] (N3) -- (N4);
    \draw[->] (N9) -- (N4);
    \draw[->] (N7) -- (N5);
    \draw[->] (N10) -- (N5);
    \draw[->] (N1) -- (N6);
    \draw[->] (N8) -- (N6);
    \draw[->] (N3) -- (N7);
    \draw[->] (N5) -- (N7);
    \draw[->] (N3) -- (N8);
    \draw[->] (N7) -- (N8);
    \draw[->] (N2) -- (N9);
    \draw[->] (N6) -- (N9);
    \draw[->] (N4) -- (N10);
    \draw[->] (N8) -- (N10);
        \end{tikzpicture}
        % \caption{Figure from the first matrix}
    \end{subfigure}
    % \hfill
    \hspace{2em}
    \begin{subfigure}{0.2\textwidth}
        \centering
        \begin{tikzpicture}[scale=0.9, transform shape]
            % Define the nodes
            \foreach \i in {1,...,10} {
                \node[circle, fill=blue!50] (N\i) at (360/10*\i:2) {\i};
            }
            
            % Draw the arrows based on the second matrix
            \draw[->] (N2) -- (N1);
            \draw[->] (N10) -- (N1);
            \draw[->] (N3) -- (N2);
            \draw[->] (N7) -- (N2);
            \draw[->] (N4) -- (N3);
            \draw[->] (N8) -- (N3);
            \draw[->] (N9) -- (N3);
            \draw[->] (N6) -- (N4);
            \draw[->] (N10) -- (N4);
            \draw[->] (N1) -- (N5);
            \draw[->] (N3) -- (N5);
            \draw[->] (N8) -- (N5);
            \draw[->] (N2) -- (N6);
            \draw[->] (N3) -- (N6);
            \draw[->] (N10) -- (N6);
            \draw[->] (N2) -- (N7);
            \draw[->] (N3) -- (N7);
            \draw[->] (N4) -- (N7);
            \draw[->] (N2) -- (N8);
            \draw[->] (N5) -- (N9);
            \draw[->] (N6) -- (N9);
            \draw[->] (N2) -- (N10);
            \draw[->] (N4) -- (N10);
            \draw[->] (N5) -- (N10);
        \end{tikzpicture}
        % \caption{Figure from the second matrix}
    \end{subfigure}
    \caption{Two networks with different densities of interconnections by varying the number of nonzero entries in the set ${ \mu_{ij} }$ for the Van der Pol network.}    \label{fig:network}
\end{figure}
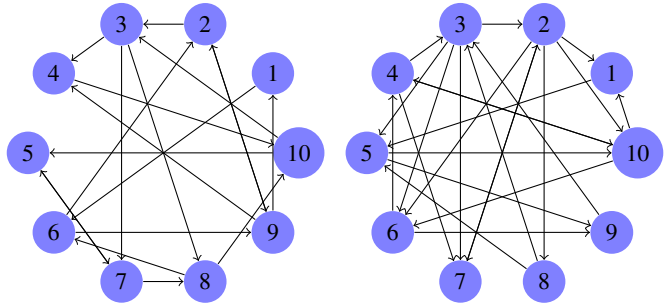

\subsection{Hyperparameters}

The domain is $[-2.5,2.5]\times [-5.5,5.5]$. The number of collocation points $N=300,000$ with a batch size of $32$ for training. We also generate $N_d=3,000$ data points and use the same batch size to evaluate the data loss (\ref{eq:L_d}). We train for a maximum of 20 epochs with a learning rate of $0.001$. 

\subsection{Results and discussions}

\begin{table*}[ht]
\centering
\begin{tabular}{ccccccc}
\toprule
\textbf{Density} & \makecell{\textbf{SOS}\\ \textbf{(scale)}}  & \textbf{Layer} & \textbf{Width} & \makecell{\textbf{Neural}\\ \textbf{(scale)}} & \makecell{\textbf{Neural}\\ \textbf{(sub-level sets)}} & \makecell{\textbf{Verification} \\\textbf{Time (s)}} \\
\midrule
3 & 0.15 & 2 & 20 & 0.27  &  $[0.21,\, 0.18,\, 0.19,\, 0.19,\, 0.24,\, 0.23,\, 0.24,\, 0.16,\, 0.19,\, 0.18]$ & 18,304\\
3 & $-$  & 2 & 30 & 0.33 & $[0.26,\, 0.22,\, 0.23,\, 0.24,\, 0.30,\, 0.29,\, 0.30,\, 0.20,\, 0.24,\, 0.22]$ & 38,521 \\
3 & $-$ & 3 & 10 & 0.29 &  $[0.24,\, 0.19,\, 0.19,\, 0.22,\, 0.26,\, 0.26,\, 0.27,\, 0.19,\, 0.21,\,  0.21]$ & 46,611\\
\midrule
4 & 0.09  & 2 & 20  & 0 & 0 & 2,952\\
4 & $-$ & 2 & 30  & 0.23 & $[0.18,\, 0.15,\, 0.16,\, 0.16,\, 0.21,\, 0.20,\, 0.21,\, 0.14,\, 0.17,\, 0.15]$ & 35,831\\
4 & $-$ & 3 & 10 & 0.21 & $[0.18,\, 0.14,\, 0.14,\, 0.16,\, 0.19,\, 0.19,\, 0.20,\, 0.14,\, 0.15,\, 0.15]$ & 35,736 \\
\bottomrule
\end{tabular}
\caption{Verification results for various neural network architectures: ``scale" denotes the largest verifiable scalar factor used to scale down the initially verified sub-level sets for individual systems. A scale equal to 0 indicates a failure to verify a stability region. Verification time includes bisection and all iterations involved in Step 2 of Section \ref{sec:method:verify}.}
\label{tab:results}
\end{table*}

We trained feedforward neural networks with $\tanh(\cdot)$ activation as Lyapunov functions for the individual subsystems using the approach proposed in \cite{liu2023towards,liu2023physics} and outlined in Section \ref{sec:method:train}. We then compositionally verified regions of attraction using the approach detailed in Section \ref{sec:method:verify}.

\begin{figure}[!htbp]
    \centering
    \includegraphics[width=0.47\textwidth]{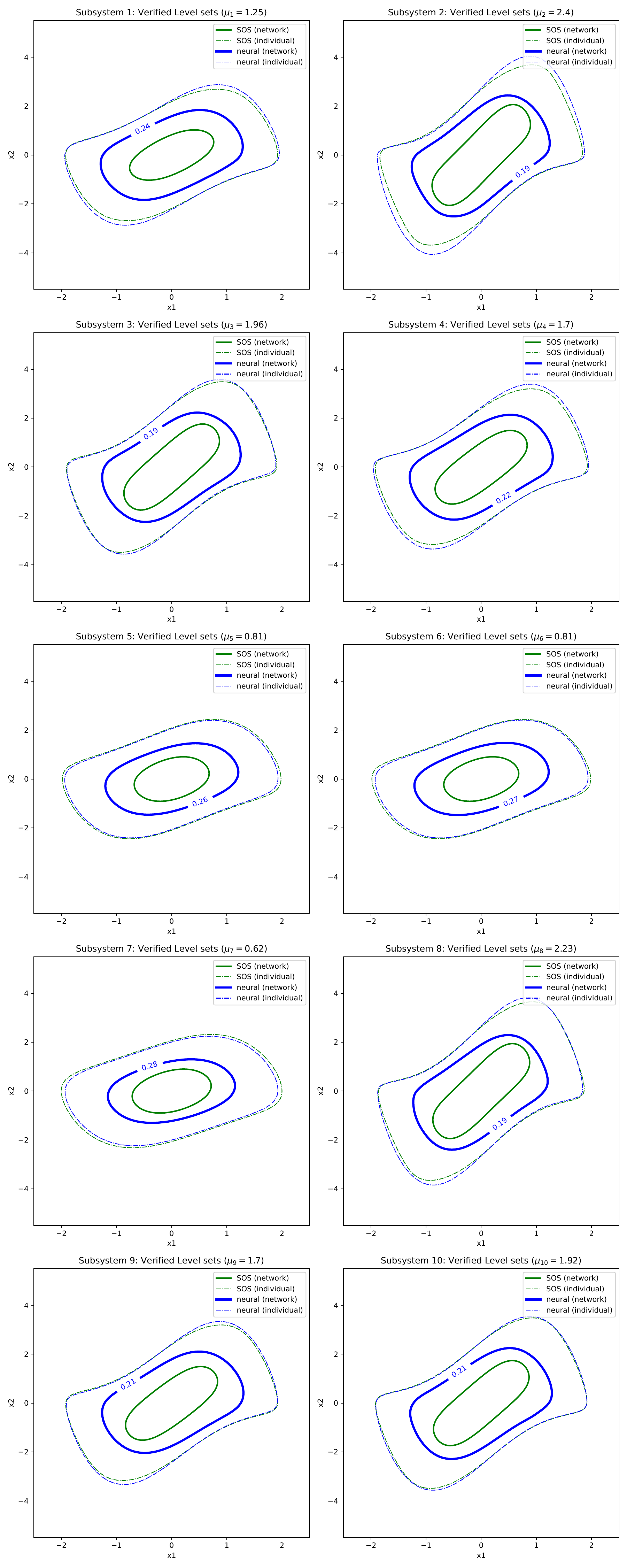}
    \caption{Neural stability analysis on a 20-dimensional interconnected system: sub-level sets of neural Lyapunov functions \cite{liu2023towards}, which define regions of attraction, are verified by the SMT solver dReal \cite{gao2013dreal} using the approach described in Section \ref{sec:method:verify}. The thick solid lines represent the regions of attraction for the interconnected system, while the thin dashed lines indicate those for individual subsystems. }    \label{fig:level_sets}
\end{figure}

For density equal 3 (network depicted on the left of Fig. \ref{fig:network}), the regions of attraction for individual subsystems, certified using SMT solvers for neural Lyapunov functions and SOS Lyapunov functions using semidefinite programming, are shown in \ref{fig:level_sets}. It can be observed that, for individual systems, the verified level sets defining regions of attraction are comparable between neural and SOS Lyapunov functions. However, the level sets for the interconnected system certified by neural Lyapunov functions are significantly better than those of the SOS Lyapunov functions. In fact, the ROA for the interconnected system certified by the SOS Lyapunov functions is even worse than that obtained using quadratic Lyapunov functions. One explanation for this is that the ``interior expanding'' approach \cite{packard2010help} expands the region of attraction in a somewhat ad hoc manner, whereas the Zubov approach solves a PDE that characterizes the region of attraction. The PDE used for training is consistent among all subsystems. This might be beneficial as they provide gains (c.f. equation (\ref{eq:ineq4})) that are more compatible among subsystems. However, this is just an intuitive explanation, and more research is likely needed before drawing firm conclusions. Moreover, for simplicity, we only considered scaled versions of the originally verified level sets for quadratic, neural, and SOS Lyapunov functions, respectively. Initializing Step 2 in Section \ref{sec:method:verify} with different initial level sets, as done in \cite{kundu2015sum}, might yield different results. Further comparisons will be conducted in future work.

We also conducted numerical experiments for the network depicted on the right of Fig. \ref{fig:network}. We summarize all the verification results in Table \ref{tab:results}. 

The computation time is reported as follows. To train a neural network on a 2020 MacBook Pro with a 2 GHz Quad-Core Intel Core i5 without any GPUs, it takes approximately 500 seconds to complete the training with the hyperparameters set as above. For verification, it takes about 60 seconds to verify a two-hidden-layer neural network with 20 neurons in each layer, roughly 200 seconds to verify a two-hidden-layer neural network with 30 neurons in each layer, and 2000 seconds to verify a three-hidden-layer neural network with 10 neurons each. We note that verification can be expedited if multiple cores are utilized, as verification in dReal leverages interval analysis, which is easily parallelizable. When verifying vector Lyapunov functions for the interconnected system, completing Step 2 in Section \ref{sec:method:verify} might require multiple iterations to reach the desired level sets. Implementing a bisection procedure to ascertain the optimal scaling down factor from the initially verified level sets can further extend computational time. The cumulative verification time for these processes is presented in Table \ref{tab:results}. Verification time was captured on an Intel(R) Xeon(R) CPU E5-2683 v4 @ 2.10GHz with 32 cores, a single CPU node through the Digital Research Alliance of Canada. The code for reproducing the example can be found at \url{https://github.com/j49liu/acc24-compositional-neural-lyapunov}. Future development for this line of research will be supported by the tool LyZNet, which targets physics-informed learning of Lyapunov functions with formal guarantees provided by SMT verification \cite{liu2024lyznet}.

\section{Conclusions}\label{sec:conclusion}

This paper introduces a compositional method to train and verify neural Lyapunov functions. Leveraging the compositional structure of interconnected nonlinear systems, we verified neural Lyapunov functions for high-dimensional systems, outperforming current SMT solvers with a monolithic method. Numerical results demonstrate that when trained using Zubov's PDE and verified compositionally with SMT solvers, neural Lyapunov functions excel over SOS Lyapunov functions obtained from semidefinite programming. Future directions will explore compositional strategies for neural Lyapunov control design.

\bibliographystyle{plain}
\bibliography{acc24}

% \newpage 

\appendix 

\subsection{Proof of Proposition \ref{prop:local}}

\begin{proof}
Define $V(x)=(V_1(x_1),\ldots,V_l(x_l))$, where $V_i(x_i) = x_i^TP_ix_i$. 

By the mean value theorem, we have
\begin{equation}
g_i(x_i) = \int_0^1 D g_i(tx_i)dt\cdot x_i,
\end{equation}
and
\begin{equation}
\begin{aligned}
    G_{ij}(x_i,x_j) = \int_0^1 D G_{ij}(tx_i,t x_j)dt \cdot \begin{bmatrix} x_i\\ x_j \end{bmatrix} \\
\end{aligned}
\end{equation}
We have 
\begin{align}
    \dot{V}_i(x_i) &= 2x_i^T P_i (f_i(x_i) + G_{ij}(x_i,x_j)) \notag\\
    & 2x_i^T P_i (A_ix_i + g_i(x_i) + \sum_{j\neq i}G_{ij}(x_i,x_j)) \notag\\
    &= -x_i^T Q_i x_i + 2 x_i^T P_i (g_i(x_i) + G_{ij}(x_i,x_j))\notag\\
 & = -x_i^T Q_i x_i + 2x_i^T\int_0^1 P_i D g_i(tx_i)dt\cdot x_i \notag\\
&\quad +  2x_i^T\int_0^1 \sum_{j\neq i} P_i D G_{ij}(tx_i,t x_j)dt \cdot \begin{bmatrix} x_i \\ x_j  \end{bmatrix} \notag\\
& \le - \frac{\lambda_{\min}(Q_i)}{\lambda_{\max}(P_i)} x_i^T P_i x_i + 2\sup_{0\le t \le 1}\norm{P_i D g_i(tx_i)}\norm{x_i}^2 \notag\\
& \quad + 2\norm{x_i} \sum_{j\neq i} \sup_{0\le t \le 1}\norm{P_i D G_{ij}(tx_i,t x_j)}\norm{\begin{bmatrix} x_i\\ x_j  \end{bmatrix}}\notag\\
& \le - \frac{\lambda_{\min}(Q_i)}{\lambda_{\max}(P_i)} x_i^T P_i x_i + 2\frac{r_{ii}+\sum_{j\neq i} r_{ij}}{\lambda_{\min}(P_i)} x_i^T P_i x_i \notag\\
& \quad + \sum_{j\neq i} \frac{r_{ij}}{\lambda_{\min}(P_j)} x_j^T P_j x_j\notag\\ 
& = \lambda_{ii} V_i(x_i) + \sum_{j\neq i} \lambda_{ij}V_j(x_j),\label{eq:dotV}
\end{align}
provided that $x_i\in \mathcal{P}_i(c_i)$ for all $i\in \set{1,\ldots,l}$.

We first show that the set $\mathcal{P}(p)$ is forward invariant. Let $\phi(t,x)$ be a solution of the interconnected system (\ref{eq:sys}) starting from $\mathcal{P}(p)$. By (\ref{eq:dotV}), if $V_i(x_i)=p_i$ for some $i$, the derivative of $V_i(x_i)$ along solutions of the interconnected system satisfies 
$$
\dot{V}_i(x_i) \le \lambda_{ii} V_i(x_i) + \sum_{j\neq i} \lambda_{ij}V_j(x_j) \le \lambda_{ii}p_i  + \sum_{j\neq i} \lambda_{ij} p_j <0.
$$
The last inequality is because $\lambda_{ij}\ge 0$ for all $j\neq i$ and $\Lambda p<0$. This proves that it is impossible for $\phi(t,x)$ to escape the sub-level set $V(x) \le p$. Hence, $\mathcal{P}(p)$ is forward invariant. 

Since $\phi(t,x)$ remains in $\mathcal{P}(p)$ for all $t\ge 0$, by inequality (\ref{eq:dotV}) and the comparison lemma (Lemma \ref{lem:compare} in the Appendix), we have $V(\phi(t,x))\le e^{\Lambda t} V(x)$ for all $t\ge 0$. By a well-known fact of non-singular $M$-matrix \cite{plemmons1977m} (see Remark \ref{rem:m-matrix}), $\Lambda$ is Hurwitz and local asymptotic stability follows. 
\end{proof}

\begin{rem}
It can be easily seen from the proof that the estimate involving the norm of $\norm{P_i D G_{ij}(x_i,x_j)}$ is not necessarily sharp. In particular, one may derive alternative bounds for the term 
$$
2\norm{x_i} \sum_{j\neq i} \sup_{0\le t \le 1}\norm{P_i D G_{ij}(tx_i,t x_j)}\norm{\begin{bmatrix} x_i\\ x_j  \end{bmatrix}}
$$ 
in the proof. We shall not pursue this in the current paper but will experiment with further improvements within the toolbox \cite{liu2024lyznet}.
\end{rem}

\subsection{Proof of Proposition \ref{prop:reach}}

\begin{proof}
The proof is straightforward. It is clear that the set $\mathcal{V}(v)$ is forward invariant for the interconnected system (\ref{eq:sys}). Furthermore, the set $\set{x_i\in\Real^{n_i}:\, V_i(x_i) \le v_i }$ is invariant for the $i$th subsystem, provided that $x_j\in \mathcal{V}_j(v_j)$. Within the set $\set{x_i\in\Real^{n_i}:\, c_{i}\le V_i(x_i) \le v_i }$, the derivative of $V_i$ along solutions of (\ref{eq:sys}) is strictly negative and bounded above by a negative real number. Hence solutions must reach $\mathcal{V}(c)$ in finite time and remain there afterwards. 
\end{proof}

\subsection{Comparison lemma for vector Lyapunov functions}

The following comparison lemma involving nonnegative matrix \cite{wazewski1950systemes,bellman1962vector} is well-known in the literature of differential inequalities and provides a useful tool for analyzing stability using vector Lyapunov functions \cite{bellman1962vector,bailey1965application}. We provide a simple self-contained proof for completeness. This linear differential inequality can also be seen as a special case of more general nonlinear comparison lemmas using quasimonotonicity (see, e.g., \cite{lakshmikantham1969differential,walter1970differential}). 

\begin{lem}[Comparison lemma\cite{bellman1962vector,wazewski1950systemes}]\label{lem:compare}
Let $\Lambda=(\lambda_{ij})\in\Real^{l\times l}$ be a matrix such that $\lambda_{ij}\ge 0$ when $i\neq j$ (i.e., with nonnegative off-diagonal entries; also known as a Metzler matrix). Assume that $V(\cdot)\in\Real^l$ satisfies the differential inequality 
\begin{equation}\label{eq:dineq}
\dot{V}(t) \le \Lambda V(t), \quad V(0)=v_0,
\end{equation}
and $W(\cdot)\in\Real^l$ solves the differential equation
\begin{equation}\label{eq:de}
\dot{W}(t) = \Lambda W(t), \quad W(0)=v_0, 
\end{equation}
for all $t\ge 0$. Then $V(t)\le W(t)$ for all $t\ge 0$. 
\end{lem}
\begin{proof}
We show that every entry of $e^{\Lambda t}$ for $t\ge 0$ is nonnegative. To see this, consider $B=\Lambda  + bI$, where $I$ is the identity matrix and $b>0$ is a constant sufficiently large such that every entry of $B$ is nonnegative. Clearly, every entry of $e^{Bt}$ is nonnegative by the definition of matrix exponential (in fact with positive diagonal entries). It follows that $e^{\Lambda t}=e^{\Lambda t + bI t} e^{-bIt} = e^{Bt} e^{-bt}$ has all nonnegative entries (and positve diagonal entries).  

We can write (\ref{eq:dineq}) as 
$$
\dot{V}(t) =  \Lambda V(t) + u(t),
$$
where $u(t)=\dot{V}(t) - \Lambda V(t) \le 0$ for all $t\ge 0$. By the general solution to a nonhomogeneous linear equation, we have
$$
V(t) = e^{\Lambda t}v_0 + \int_0^t e^{\Lambda(t-s)}u(s)ds \le e^{\Lambda t}v_0 = W(t),
$$
where the inequality holds because all entries of $e^{\Lambda(t-s)}$ are nonnegative and $u(s)$ is a nonpositive vector. 
\end{proof}

\end{document}